\documentclass[pra,superscriptaddress,showpacs,twocolumn,10pt]{revtex4}

\usepackage{amsmath,amsthm}
\usepackage{amssymb,latexsym}
\usepackage{float}
\usepackage{graphicx}

\newtheorem{theorem}{Theorem}

\newtheorem{lemma}[theorem]{Lemma}

\theoremstyle{definition}

\floatstyle{ruled}
\newfloat{algorithm}{htb}{loa}
\floatname{algorithm}{Algorithm}

\newcommand{\qedsymb}{\hfill{\rule{2mm}{2mm}}}
\renewenvironment{proof}[1][]{\begin{trivlist} % Note changed to renewenvironment by Ben because I loaded amsthm package.
\item[\hspace{\labelsep}{\bf\noindent Proof#1:\/}] }{\qedsymb\end{trivlist}}

\newcommand{\be}{\begin{equation}}
\newcommand{\ee}{\end{equation}}

\newcommand{\bea}{\begin{eqnarray}}
\newcommand{\eea}{\end{eqnarray}}

\newcommand{\ba}{\begin{array}}
\newcommand{\ea}{\end{array}}

\newcommand{\bc}{\begin{cases}}
\newcommand{\ec}{\end{cases}}

\newcommand{\beba}{\begin{equation}\begin{array}{lll}}
\newcommand{\eeea}{\end{array}\end{equation}}

\newcommand{\bp}{\begin{pmatrix}}
\newcommand{\ep}{\end{pmatrix}}

\newcommand{\bit}{\begin{itemize}}
\newcommand{\eit}{\end{itemize}}

\newcommand{\ben}{\begin{enumerate}}
\newcommand{\een}{\end{enumerate}}

\newcommand{\ket}[1]{|#1\rangle}

\newcommand{\bracket}[2]{\langle #1|#2\rangle}
\newcommand{\ketbra}[2]{|#1\rangle \langle #2 |}

\newcommand{\ip}[2]{\langle #1, #2\rangle}

\newcommand{\ssll}[2]{\sum\limits_{#1}^{#2}}

\newcommand{\vx}{\vec{x}}
\newcommand{\vy}{\vec{y}}
\newcommand{\vz}{\vec{z}}

\newcommand{\MNN}{\mathbb{M}_{N,N}}

\newcommand{\UN}{\mathbb{U}_N}
\newcommand{\ON}{\mathbb{O}_N}
\newcommand{\CN}{\mathbb{C}^N}
\newcommand{\Cn}{\mathcal{C}_n}
\newcommand{\Pn}{\mathcal{P}_n}

\newcommand{\CNN}{\mathbb{C}^{N \times N}}

\newcommand{\ME}{\ket{\Phi^+_N}}

\renewcommand{\S}{\mathcal{S}}
\newcommand{\T}{\mathcal{T}}
\renewcommand{\span}{\mathrm{span}}
\newcommand{\WK}{\mathcal{W}_K}

\newcommand{\tr}{\mathrm{tr}}
\newcommand{\diag}{\mathrm{diag}}
\newcommand{\supp}{\mathrm{supp}}

\begin{document}

\title{Property testing of unitary operators}

\author{Guoming Wang}
\email{gmwang@eecs.berkeley.edu}
\affiliation{Computer Science Division, University of California Berkeley, Berkeley, California 94720, USA}

\date{\today}

\begin{abstract}
In this paper, we systematically study property  testing of unitary operators. We first introduce a distance measure that reflects the average difference between unitary operators. Then we show that, with respect to this distance measure, the orthogonal group, quantum juntas (i.e. unitary operators that only nontrivially act on a few qubits of the system) and Clifford group can be all efficiently tested. In fact, their testing algorithms have query complexities independent of the system's size and have only one-sided error. Then we give an algorithm that tests any finite subset of the unitary group, and demonstrate an application of this algorithm to the permutation group. This algorithm also has one-sided error and polynomial query complexity, but it is unknown whether it can be efficiently implemented in general. 
\end{abstract}


\maketitle

\section{Introduction}

%difficulty of tomography 
Characterizing the dynamical behavior of complex quantum systems is an important yet daunting task. The standard approach known as quantum process tomography (QPT) \cite{NC00,CN97,PCZ97} can provide full information about the quantum process, but it consumes a huge amount of resource. Namely, in order to fully determine a quantum operation acting on a system consisting of $n$ qubits, QPT needs to use $\Theta(16^n)$ observables to estimate all the parameters necessary to describe this operation. Even if this operation is known to be unitary, it still needs $\Theta(4^n)$ observables. Although many improvements and variants of QPT have been proposed \cite{AP01,AP03,ABJ+03,EAZ05,ML06,ESM+07,ML07,BPP08,DCEL09,
BNWK09,BPP09,SLP10,DLF+10,SKM+11,SBLP11,MGE11}, in general the resource consumption of this approach still grows quickly as the system becomes large.

%introduce property testing
On the other hand, in many situations we might not need to fully determine the quantum operation, but merely wish to know whether it satisfies certain property or is far from having this property (assuming it is one of the two cases). For example, given a quantum machine acting on an $n$-qubit system, it is natural to ask whether it only nontrivially acts on a few qubits, or it has non-negligible effect on every qubit. Similar questions have been raised and studied in the classical situation. For example, given a boolean function (or a graph) as an oracle, we want to know whether the function is linear (or the graph is connected) or is far from any of such functions (or graphs) with respect to some reasonable metric, by making only a few queries to the oracle. This problem is usually called \textit{property testing}\cite{BLR93,RS96,GGR98}. It has been extensively studied in computer science and has wide applications such as probabilistically checkable proofs (PCP)\cite{ALM+98}. Surprisingly, many properties of functions and graphs are found to be testable with very few queries. Sometimes the query complexity is even independent of the input's size.

Given these facts, one may naturally wonder whether the less ambitious goal of property testing of quantum states or operations would lead to a dramatic decrease in resource consumption. Several previous results indicated that it is indeed the case\cite{MO09,Low09,HM10,SCP11,FL11}. For example, the separability of multipartite states and operations can be tested with query complexity independent of the system's size\citep{HM10}.
In this paper, we will continue this line of research and focus on studying property testing of unitary operators. (The reader should not confuse our work with \textit{quantum property testing} \cite{BFNR02}, which concerns about the testing of classical objects with quantum algorithms. Here we are interested in the testing of quantum objects themselves.) We first introduce a normalized distance measure that quantifies the average difference between unitary operators. Then we show that, with respect to this distance measure, the orthogonal group, quantum juntas (i.e. unitary operators that only nontrivially act on a few qubits of the system) and Clifford group can be all efficiently tested. In fact, their testing algorithms have query complexities independent of the system's size and also have only one-sided error. Next, we give a general algorithm that tests any finite subset of the unitary group, and demonstrate an application of this algorithm to the permutation group. This algorithm also has one-sided error and polynomial query complexity, but we do not know whether it can be efficiently implemented in general. 

%organization
The remainder of this paper is organized as follows.
In section 2, we introduce the definitions, notations and tools used in this paper.
Then, in sections 3, 4 and 5, we study the testing of orthogonal group, quantum juntas and Clifford group respectively. In section 6, we present an algorithm that tests any finite subset of the unitary group, and then exhibits its application to the permutation group. Finally, section 7 concludes this paper.

\section{Preliminaries}

\subsection{Definitions and Notations}
In this section, we introduce the definitions and notations used in this papers.

Let $n \ge 1$ and $N=2^n$. We use $\MNN$ to denote the set of linear operators from $\CN$ to $\CN$ (given a fixed basis for $\CN$, they are represented by $N \times N$ matrices with complex entries), and use $\UN=\{U\in \MNN: UU^{\dagger}=U^{\dagger}U=I\}$
to denote the set of $N$-dimensional unitary operators.
We are going to regard $\MNN$ as a Hilbert space 
equipped with the Hilbert-Schmidt inner product
\beba
\ip{A}{B}=\tr(A^{\dagger}B).
\label{eq:ip}
\eeea
This inner product induces the Hilbert-Schmidt
(or Frobenius) norm for $A=(a_{i,j})_{i,j=1}^N$
\beba
\|A\|&=&\sqrt{\tr(A^{\dagger}A)}\\
&=&\sqrt{\ssll{i=1}{N} \ssll{j=1}{N} |a_{i,j}|^2}.\\
\eeea
This norm further induces the following metric
\beba
d(A,B)=\|A-B\|.
\eeea
But this metric might be not good for comparing unitary operators, since in general we have $d(U,V) \neq d(e^{i\theta}U,V)$ for $\theta \in (0,2\pi)$, although $U$ and $e^{i\theta}U$ are usually considered as the same operation since they are equivalent up to a global phase. To overcome this problem, we introduce another distance measure as follows. First, we define an equivalence relation between linear operators as follows: $A \sim B$ if and only if $A=e^{i\theta} B$ for some $\theta \in [0,2\pi)$. Then for any $A \in \MNN$, define $[A]=\{B:~A \sim B\}$. The distance between $A$ and $B$ is given by 
\beba
D(A,B)&=&\min\limits_{C \in [A], D \in [B]}
\dfrac{1}{\sqrt{2N}}\|C-D\|\\
&=&\min\limits_{\theta \in [0,2\pi)}\dfrac{1}{\sqrt{2N}}\|e^{i\theta}A-B\|.
\label{eq:D1}
\eeea
More generally, for any $\S \subseteq \MNN$, define $[\S]=\cup_{A \in \S} [A]$.
And the distance between two sets $\S$ and $\T$ is given by
\beba
D(\S,\T)&=\min\limits_{A \in \S, B \in \T} D(A,B)\\
&=\min\limits_{A \in [\S], B \in [\T]} \dfrac{1}{\sqrt{2N}} \|A-B\|.
\label{eq:D2}
\eeea
It can be easily checked that 
\ben
\item $D(A,B) \ge 0$, and the equality holds if and only if $A \sim B$.  
\item $D(A,B)=D(B,A)$.
\item $D(A,B)+D(B,C) \ge D(A,C)$.
\een
Besides, for unitary operators, $D$ has the following nice properties:
\ben
\item $D(U,V) \le 1$.
\item $D(UV_1,UV_2)=D(V_1,V_2)$.
\item $D(U \otimes V_1,U \otimes V_2)=D(V_1,V_2)$.
\item $D(U_1V_1,U_2V_2) \le D(U_1,U_2)+D(V_1,V_2)$.
This is a consequence of $D(U_1V_1,U_2V_2) \le D(U_1V_1,U_2V_1)+D(U_2V_1,U_2V_2)$
and property 2.
\item $D(U_1 \otimes V_1, U_2 \otimes V_2)
\le D(U_1,U_2)+D(V_1,V_2)$.
This is a consequence of $D(U_1 \otimes V_1,U_2 \otimes V_2) \le D(U_1 \otimes V_1,U_2 \otimes V_1)+D(U_2 \otimes V_1,U_2 \otimes V_2)$ and property 3.
\een
Thus $D$ is a normalized distance measure that reflects the average difference between unitary operators. In addition, the following relation between $D(U,V)$ and $\ip{U}{V}$ would be very useful:
\beba
D^2(U,V)=1-\dfrac{1}{N}|\ip{U}{V}|.
\label{eq:dip}
\eeea

\subsection{Our Question}
%definition about property testing
The task of property testing is typically described as follows. Suppose some unknown object, such as a graph or a boolean function, is given as an oracle which can be queried locally many times. Our goal is to determine whether this object has certain global property or is far from having this property, by making as few queries as possible. 	

Formally, let $\Omega$ be a predetermined set from which the object
is chosen. $\Omega$ should be also equipped with a distance
measure $d$. A property is a subset $\S \subset \Omega$. For any $A \in \Omega$, if $A \in \S$, then we say that $A$ has property $S$; otherwise, if $d(A,\S) \ge \epsilon$, i.e. $d(A,B)\ge \epsilon$
for any $B \in \S$, then we say that $A$ is  $\epsilon$-far from property $\S$. An algorithm $\epsilon$-tests property $\S$ if for any input $A \in \Omega$, 
\bit
\item if $A$ has property $\S$, then the algorithm accepts $A$ with probability
at least $2/3$;
\item if $A$ is $\epsilon$-far from property $\S$, then the algorithm accepts $A$ with probability at most $1/3$.
\eit
Besides, if the algorithm makes at most $q(|\Omega|,\epsilon)$ queries to the oracle, then we say that it has query complexity $O(q(|\Omega|,\epsilon))$.
A testing algorithm would be very efficient if its query complexity  depends only on $\epsilon$
but not on $|\Omega|$.

In this paper, we will study the problem of property testing of unitary operators. In our case, $\Omega=\UN$ and we use $D$ defined 
as Eqs.(\ref{eq:D1}) and (\ref{eq:D2}) as the distance measure. 
However, we need to slightly change the definition of having a property and being far from a property as follows: let $\S \subset \UN$ be a property. We say $U \in \UN$ has property
$\S$ if $D(U,\S)=0$, i.e. $U \in [\S]$; otherwise, we say $U$ is $\epsilon$-far from property $S$ if $D(U,\S) \ge \epsilon$, i.e. $D(U,V) \ge \epsilon$, $\forall V \in \S$. Our input is a blackbox implementing some $U \in \UN$ which can be 
accessed as follows: first, we prepare some state 
$\ket{\psi_{AB}}$, where $A,B$ are its two subsystems such that $\dim(A)=N$ and $B$ is some auxiliary system; then we apply $U$ on the $A$ subsystem, obtain $\ket{\psi_{AB}'}=(U \otimes I)\ket{\psi_{AB}}$; finally we perform some measurement on $\ket{\psi_{AB}'}$ and get information about $U$. In certain cases, we are also allowed to access a blackbox implementing $U^{\dagger}$. Our goal is that given any $\S \subseteq \UN$, find an algorithm that $\epsilon$-tests $\S$ with the minimal query complexity.

\subsection{Useful tools}

The following tools will be very useful for our work.

\subsubsection{Choi-Jamio\l kowski Isomorphism}

The Choi-Jamio\l kowski isomorphism \cite{Jam72,Cho75} states that there is a duality between 
quantum channels and quantum states. In particular, there exists an isomorphism 
between unitary operators in $\UN$ and pure states in $\CNN$.
Specifically, let 
\beba
\ME=\dfrac{1}{\sqrt{N}}\ssll{i=1}{N}\ket{i}\ket{i}
\eeea
be the $N$-dimensional Bell state.
For any $A \in \MNN$, define
\beba
\ket{v(A)} = (A \otimes I)\ME,
\eeea
where $A$ is applied to the first subsystem.
Then we have
\beba
\bracket{v(A)}{v(B)}&=&\dfrac{1}{N}\ip{A}{B}.\\
\eeea
In particular, for any $U,V \in \UN$,  we have
\beba
\bracket{v(U)}{v(V)} &=& \dfrac{1}{N}\ip{U}{V},\\
\label{eq:cjuv}
\eeea
and
\beba
\|\ket{v(U)}\| = \|\ket{v(V)}\| = 1.
\eeea
So the angle between $\ket{v(U)}$ and $\ket{v(V)}$ faithfully reflects the ``angle'' between $U$ and $V$ with respect to the Hilbert-Schimdt product. And if we perform the projective measurement $\{\ketbra{v(V)}{v(V)}, I-\ketbra{v(V)}{v(V)}\}$ on the state $\ket{v(U)}$, then the probability of obtaining the outcome corresponding to $\ketbra{v(V)}{v(V)}$ is $|\langle U,V\rangle|^2/N^2$.

\subsubsection{Singular Value Decomposition}
Suppose $A$ has the singular value decomposition
\beba
A=V_1 \Sigma V_2,
\label{eq:svd}
\eeea
where $\Sigma=\diag(\sigma_1,\sigma_2,\dots,\sigma_N)$
with $\sigma_i \ge 0$, and $V_1,V_2 \in \UN$. Then we have
\beba
\|A\|=\sqrt{\ssll{i=1}{N}\sigma_i^2}.
\eeea
If $A$ is unitary, then its singular values $\sigma_i=1$. The following lemma shows that the converse is also true in an approximate sense:
(In what follows, when we write $A \le B$ for $A,B \in \MNN$, we means that $B-A$ is semidefinite positive.)

\begin{lemma}
If $A \in \MNN$ satisfies $A^{\dagger}A \le I$
and $\|A\|^2 \ge N(1-\epsilon)$,
then $D(A,\UN) \le \sqrt{\epsilon/2}.$
\label{lem:svd}
\end{lemma}
\begin{proof}
Suppose $A$ has the singular value decomposition
as Eq.(\ref{eq:svd}). Note that the condition 
$A^{\dagger}A \le I$ is equivalent to 
\beba
\sigma_i \le 1,~~~ \forall i.
\label{eq:sigma1}
\eeea
And the condition
$\|A\|^2 \ge N(1-\epsilon)$ is equivalent to
\beba
\ssll{i=1}{N} \sigma_i^2 \ge N(1-\epsilon).
\label{eq:sigma2}
\eeea
Define $U=V_1V_2 \in \UN$. Then
\beba
D^2(U,A) &\le& \frac{1}{2N}\|U-A\|^2\\
&=&\frac{1}{2N}\|V_1(I-\Sigma)V_2\|^2\\
&=&\frac{1}{2N}\|I-\Sigma\|^2\\
&=&\frac{1}{2N}\ssll{i=1}{N} (1-\sigma_i)^2\\
&\le &\frac{1}{2N}\ssll{i=1}{N} (1-\sigma_i^2) \\
&\le &\epsilon/2,
\eeea
which implies $D(U,A) \le \sqrt{\epsilon/2}.$
\end{proof}

\subsubsection{Pauli Decomposition}
Let
\beba
\sigma_{0}=\bp
1 & 0\\
0 & 1\\
\ep=I,&
\sigma_{1}=\bp
0 & 1\\
1 & 0\\
\ep=X,\\
\sigma_{2}=\bp
0 & -i\\
i & 0\\
\ep=Y,&
\sigma_{3}=\bp
1 & 0\\
0 & -1\\
\ep=Z
\eeea
be the Pauli operators. And let $\mathbb{Z}_4=\{0,1,2,3\}$, 
$\mathbb{Z}^n_4=\{0,1,2,3\}^n$.
For any $\vec{x}=(x_1,x_2,\dots,x_n) \in \mathbb{Z}^n_4$,
define
\beba
\sigma_{\vec{x}}=\sigma_{x_1} \otimes \sigma_{x_2} \otimes \dots \otimes \sigma_{x_n}.
\eeea
Then $\{\frac{1}{\sqrt{N}}\sigma_{\vx}\}_{x \in \mathbb{Z}_4^n}$
form an orthonormal basis for $\MNN$ with respect to the Hilbert-Schmidt inner product. So we can write 
any $A \in \MNN$ as 
\beba
A = \sum\limits_{\vx \in \mathbb{Z}_4^n}{\mu_{\vx}(A)\sigma_{\vx}},
\label{eq:pd}
\eeea
where 
\beba
\mu_{\vx}(A)= \dfrac{1}{N}\ip{A}{\sigma_{\vx}}.
\label{eq:pdcof}
\eeea
It can be easily checked that
\beba
\|A\|^2=N\ssll{\vx \in \mathbb{Z}_4^n}{}{|\mu_{\vx}(A)|^2}.
\label{eq:paulinorm}
\eeea
In particular, any $U \in \UN$ satisfies
\beba
\frac{1}{N}\|U\|^2=\ssll{\vx \in \mathbb{Z}_4^n}{}{|\mu_{\vx}(U)|^2}=1.
\eeea

By the Choi-Jamio\l kowski isomorphism, $\{\ket{v(\sigma_{\vx})}\}_{x \in \mathbb{Z}_4^n}$ also form an orthonormal basis for $\CNN$. 
For any $A \in \MNN$, we have
\beba
\ket{v(A)} = \sum\limits_{\vx \in \mathbb{Z}_4^n}{\mu_{\vx}(A)\ket{v(x)}}.
\label{eq:pdst}
\eeea
Hence, if we measure the state $\ket{v(A)}$ (assuming it is normalized) in the basis $\{\ket{v(\sigma_{\vx})}: x \in \mathbb{Z}_4^n\}$, we get the outcome $\vx$ with probability $|\mu_{\vx}(A)|^2$.

Finally, we define $\oplus:\mathbb{Z}_4^n \times \mathbb{Z}_4^n \to \mathbb{Z}_4^n$
and $\odot:\mathbb{Z}_4^n \times \mathbb{Z}_4^n \to \mathbb{Z}_4$ such that for any $\vx,\vy \in \mathbb{Z}_4^n$,  
\beba
\sigma_{\vx}\sigma_{\vy}=
i^{ \vx \odot \vy} \sigma_{\vx \oplus \vy}.
\eeea
Then
\ben
\item $\vx \oplus \vy=\vy \oplus \vx$.
\item If $\sigma_{\vx}$ and $\sigma_{\vy}$
commute, then $i^{\vx \odot \vy}=i^{\vy\odot\vx}$; otherwise, $i^{\vx \odot \vy}=-i^{\vy\odot\vx}$.
\een

\section{Testing Orthogonal Group}

Let us begin with the testing of
orthogonal group
\beba
\ON &=& \{U \in \UN: U^TU=UU^T=I\}\\
&=& \{U \in \UN: U=U^*\},
\eeea
which is an important subgroup of $\UN$
that has wide applications. Our first result is

\begin{theorem}
The orthogonal group $\ON$ can be $\epsilon$-tested with query
complexity $O(1/\epsilon^2)$.
\label{thm:ort}
\end{theorem}
\begin{proof}
Our basic idea is to show that: 
(1) if $UU^T$ is close to $I$, then $U$ is close to $\ON$;
(2) we can test the proximity between $UU^T$
and $I$ very efficiently. Now we provide the
details.

First, we prove 
\begin{lemma}
If $U \in \UN$ satisfies $|\tr(UU^T)|\ge N(1-\delta)$,
then $D(U,\ON) \le \sqrt{\delta}$.
\label{lem:ort}
\end{lemma}
\begin{proof}
Without loss of generality, we can assume that
$\tr(UU^T)$ is real, because, if otherwise,
we can replace $U$ by some $U_1 \in [U]$ such
that $\tr(U_1U_1^T)=|\tr(UU^T)|$
and $D(U,\ON)=D(U_1,\ON)$.

Let $A=(U+U^*)/2$. Then $A$ is real and 
\beba
A^{\dagger}A&=(U^{\dagger}+U^T)(U+U^*)/4 \\
&=(2I+U^{\dagger}U^*+U^TU)/4\\
&\le I.
\eeea
Moreover,
\beba
\|A\|^2&=& \tr(A^{\dagger}A)\\
&=&\tr(2I+U^{\dagger}U^*+U^TU)/4\\
&\ge & N(1-\delta/2).
\eeea
Suppose $A$ has the singular value decomposition
$A=V_1 \Sigma V_2$, where $V_1,V_2 \in \ON$
since $A$ is real. Then define $V=V_1V_2 \in \ON$.
Then similar to the proof of lemma \ref{lem:svd},
we can show
\beba
D(A,V) \le 	\sqrt{\delta}/2.
\eeea
Meanwhile, we have
\beba
\|U-A\|^2&=& \|(U-U^*)/2\|^2\\
&=&\tr((U^{\dagger}-U^T)(U-U^*))/4\\
&=& \tr(2I-U^{\dagger}U^*-U^TU)/4\\
&\le & N\delta/2,
\eeea
which implies
\beba
D(U,A) \le \sqrt{\delta}/2.
\eeea
Therefore,
\beba
D(U,\ON) &\le D(U,V)\\ 
&\le D(U,A)+D(A,V)\\
&\le \sqrt{\delta}.
\eeea
\end{proof}

Now consider the following testing algorithm:\\

\begin{algorithm}
\caption{~~~Testing $\ON$}
\begin{tabular}{p{0.06 \textwidth}p{0.4 \textwidth}}
\textbf{Input:} & $U \in \UN$ is given as a blackbox. $\epsilon>0$ is a proximity parameter.\\
\textbf{Steps:} & \\
& 1. Repeat the following procedure $O(1/\epsilon^2)$ times
\bit
\item Prepare the state $\ket{\Phi^+_N}$; 
\item Apply the operation $U \otimes U$ to it;
\item Perform the measurement $\{\ketbra{\Phi^+_N}{\Phi^+_N},
I-\ketbra{\Phi^+_N}{\Phi^+_N}\}$ on the state. If the outcome
corresponds to $\ketbra{\Phi^+_N}{\Phi^+_N}$, then this iteration
is successful.
\eit
\\
& 2. Accept if and only if all iterations are successful.\\
\end{tabular}
\label{alg:ort}
\end{algorithm}

Note that
\beba
(I \otimes A) \ket{\Phi^+_N} = (A^T \otimes I) \ket{\Phi^+_N},~~
\forall A \in \MNN.
\eeea
So
\beba
(U \otimes U) \ket{\Phi^+_N} = (UU^T \otimes I) \ket{\Phi^+_N}.
\eeea
The probability that an iteration is successful is
\beba
|\langle \Phi^+_N |(U \otimes U)| \Phi^+_N \rangle|^2 
&=& |\langle \Phi^+_N |(UU^T \otimes I)| \Phi^+_N \rangle|^2\\
&=& |\tr(UU^T)|^2/N^2\\
&=& (1-\gamma)^2,
\eeea
assuming $|\tr(UU^T)|=N(1-\gamma)$ for some $0 \le \gamma \le 1$.
Suppose we run $1/\epsilon^2$ iterations in step 1. Then the algorithm accepts $U$ with probability 
$(1-\gamma)^{2/\epsilon^2} \le e^{-2\gamma/\epsilon^2}$,
which is smaller than $1/3$
for any $\gamma>\epsilon^2$. Hence if the algorithm accepts accepts $U$ with probability
at least $1/3$, then $\gamma \le \epsilon^2$
and by lemma \ref{lem:ort} we have $D(U,\ON) \le \epsilon$. On the other hand, for any $U \in [\ON]$, 
obviously every iteration is successful and the algorithm always accepts it. So this algorithm $\epsilon$-tests $\ON$ by making $O(1/\epsilon^2)$ queries to $U$.
\end{proof}

Note that $\ket{\Phi^+_N}=\ket{v(\sigma_{\vec{0}})}$ can be efficiently prepared, and 
the measurement $\{\ketbra{\Phi^+_N}{\Phi^+_N}, I-\ketbra{\Phi^+_N}{\Phi^+_N}\}$
can be simulated by a finer measurement in the basis
$\{ \ket{v(\sigma_{\vx})}\}_{\vx \in \mathbb{Z}_4^n}$.
Hence, algorithm \ref{alg:ort} can be efficiently implemented.

\section{Testing Quantum Juntas}

%define quantum junta
Given a unitary operator $U$ acting on an $n$-qubit system, we might want to know if it only nontrivially acts on at most $k$ of the $n$ qubits. Formally, let $[n]=\{1,2,\dots,n\}$ be the set of indices of the qubits. For any $T \subseteq [n]$, we also use $T$ to denote the subsystem composed of qubits whose indices are in $T$.
We say that $U \in \UN$ only nontrivially acts on subsystem $T$ if $U=V_T \otimes I_{T^c}$ for some $V \in \mathbb{U}_{2^{|T|}}$, where $V_T$
and $I_{T^c}$ indicate that $V$ and $I$ act on the subsystem $T$ and $T^c=[n]\setminus T$ respectively. The set of quantum $k$-juntas is defined as
\beba
k\textrm{-Juntas} &=& \{U \in \UN: \exists T \subseteq [n], |T|=k, ~\textrm{s.t.}\\
&&U=V_T \otimes I_{T^c} \textrm{~for~some~} V \in \mathbb{U}_{2^k}\}.
\eeea

\begin{theorem}
$k$-Juntas can be $\epsilon$-tested with query
complexity $O(k \log (k/\epsilon)/ \epsilon^2)$.
\end{theorem}
\begin{proof}
%intuition
Our basic idea is to consider the Pauli decomposition given by Eqs.(\ref{eq:pd}-\ref{eq:pdst}). For any $\vx \in \mathbb{Z}_4^n$,
let $\supp(\vx)=\{i:~x_i \neq 0\}$ and $|\vx|=|\supp(\vx)|$. 
Note that if $U=V_T \otimes I_{T^c}$ for some $T \subseteq [n]$, then $\mu_{\vx}(U)=0$ for any $\vx$ with $\supp(\vx) \not\subseteq T$. So for any $k$-junta $U$, if we measure the state $\ket{v(U)}$ in the basis $\{\ket{v(\sigma_{\vx})}\}_{x \in \mathbb{Z}_4^n}$, then we only obtain outcome $\vx$ satisfying $\supp(\vx) \subseteq T$ for some $T \subseteq [n]$ with $|T|=k$. The difficult part is to prove the converse is also true in the approximate sense. Namely, if we obtain outcome $\vx$ satisfying the same condition for sufficiently high probability, then $U$ is close to a $k$-junta.

%give the algorithm
Now we give the details. Consider the following testing algorithm:

\begin{algorithm}
\caption{~~~Testing $k$-Juntas}
\begin{tabular}{p{0.15 \columnwidth}p{0.75 \columnwidth}}
\textbf{Input:} & $U \in \UN$ is given as a blackbox.
$\epsilon>0$ is a proximity parameter. \\
\textbf{Steps:} & \\
& 1. Let $W=\varnothing$. \\
& 2. Repeat the following procedure $O(k \log ({k}/{\epsilon})/\epsilon^2)$ times:
\bit
\item Measurement the state $\ket{v(U)}$ in the basis
$\{\ket{v(\sigma_{\vx})}\}_{x \in \mathbb{Z}_4^n}$. Suppose we get the outcome $\vx$. 
\item Update $W \to W \cup \supp(\vx)$.
\item If $|W|>k$, then reject and quit. 
\eit \\
& 3. If none of the iterations rejects, then accept.\\
\end{tabular}
\label{alg:junta}
\end{algorithm}

%analysis
Obviously, this algorithm accepts all $k$-juntas. It remains to show that if it accepts $U$ with probability at least $1/3$, then $D(U,k\textrm{-Juntas}) \le \epsilon.$ To prove this, it is enough to prove the following statements: 

Let $\mathbb{Z}_4^T=\{\vx \in \mathbb{Z}_4^n: \supp(\vx) \subseteq T\}$ for any $T \subseteq [n]$. Then
\ben
\item if $\sum\limits_{\vx \in \mathbb{Z}_4^T} |\mu_{\vx}(U)|^2 \le 1-\epsilon^2/4$ for any $T \subseteq [n]$, $|T| \le k$, then the algorithm accepts $U$ with probability at most $1/3$;
\item if 
$\sum\limits_{\vx \in \mathbb{Z}_4^T} |\mu_{\vx}(U)|^2
 \ge 1-\delta$ 
for some $T \subseteq [n]$, then there exists $\tilde{V} \in \mathbb{U}_{2^{|T|}}$ such that
$D(U,\tilde{V}_T \otimes I_{T^c}) \le 2\sqrt{\delta}.$ 
\een
The desired result follows immediately from the two statements.

To prove the first statement, consider the following classical game. Repeat the following procedure $m$ times: each time, we randomly sample a string from $\mathbb{Z}_4^n$ such that any $\vx \in \mathbb{Z}_4^n$ is chosen with probability $p_{\vx}=|\mu_{\vx}(U)|^2$. Let $\vx_1, \vx_2, \dots, \vx_m$ be the samples. We win the game if and only if 
\beba
|\bigcup\limits_{j=1}^m \supp(\vx_j)|\le k.
\eeea 
If for any $T \subseteq [n]$, $|T| \le k$, 
\beba
\Pr[\supp(\vx) \subseteq T]=\sum\limits_{\vx \in \mathbb{Z}_4^T} p_{\vx} \le 1-\delta,
\eeea
then what is the maximal probability of winning the game? 

Of course, we can give a simple upper bound $\binom{n}{k}(1-\delta)^m$, since there are $\binom{n}{k}$ possibilities of $W=\bigcup\limits_{j=1}^m \supp(\vx_j)$ such that $|W| \le k$, and when $W$ is fixed, $\Pr[\supp(\vx_j) \subseteq W]<1-\delta$ for each $j$. However, this bound is not good enough since it depends on $n$. To get a better bound, we introduce the following concept: given a sequence 
of samples $\vx_1,\vx_2,\dots,\vx_m$, we say that $j \in [m]$ is a \textit{support-defining} position if 
\beba
\supp(\vx_j) \not\subseteq \bigcup\limits_{i=1}^{j-1} \supp(\vx_i).
\eeea
Note that if $|W| \le k$, then there can be at most $k$ support-defining positions. Consequently, we can find $1 \le j_1<j_2<\dots<j_k \le m$ such that $W=\bigcup\limits_{i=1}^{k} \supp(\vx_{j_i}).$
No matter what $W$ is, for any $j \in [m]\setminus\{j_1,j_2,\dots,j_k\}$, we always have
$\Pr[\supp(\vx_j) \subseteq W] \le 1-\delta.$
Since there can be at most $\binom{m}{k}$ choices of $j_1,j_2,\dots,j_k$, by union bound we get
\beba
\Pr[|W| \le k] &\le \binom{m}{k}(1-\delta)^{m-k}\\
& \le m^k (1-\delta)^{m-k}\\
& \le e^{k \log m-\delta (m-k)}.
\eeea
By choosing $m=O(k\log (k/\delta) /\delta)$,
we can make this probability smaller than $1/3$.
Setting $\delta=\epsilon^2/4$ yields statement 1.

Now we prove the second statement. Let
\beba
A&=&\sum\limits_{\vx \in \mathbb{Z}_4^T} \mu_{\vx}(U) \sigma_{\vx}
=\tilde{A}_T \otimes I_{T^c},\\
B&=&U-A=\sum\limits_{\vx \not\in \mathbb{Z}_4^T} \mu_{\vx}(U) \sigma_{\vx}.
\eeea
Then we have
\beba
\|A\|^2=N\sum\limits_{\vx \in \mathbb{Z}_4^T} |\mu_{\vx}(U)|^2
\ge N(1-\delta),\\
\|B\|^2=N\sum\limits_{\vx \not\in \mathbb{Z}_4^T} |\mu_{\vx}(U)|^2
\le N\delta,
\eeea
which implies
\bea
\|\tilde{A}\|^2 &\ge & 2^{|T|}(1-\delta).
\label{eq:tildeAnorm}\\
\|B\| &\le & \sqrt{N\delta}.
\eea
Furthermore, let $\ket{\psi}$ be an arbitrary state on subsystem $T$ and let $\rho$ be the uniformly mixed state on subsystem $T^c$. Then we have
\beba
1&=&\tr(U^{\dagger}U (\ketbra{\psi}{\psi} \otimes \rho))\\
&=& \tr (A^{\dagger}A (\ketbra{\psi}{\psi} \otimes \rho))
+\tr( B^{\dagger}B (\ketbra{\psi}{\psi} \otimes \rho))\\
&&+\tr (A^{\dagger}B (\ketbra{\psi}{\psi} \otimes \rho))
+\tr (B^{\dagger}A (\ketbra{\psi}{\psi} \otimes \rho))\\
&=& \tr (\tilde{A}^{\dagger}\tilde{A} (\ketbra{\psi}{\psi})
+\tr( B^{\dagger}B (\ketbra{\psi}{\psi} \otimes \rho))\\
& \ge & \tr (\tilde{A}^{\dagger}\tilde{A} (\ketbra{\psi}{\psi}),
\eeea
where in the third step we use the fact
\beba
\mu_{\vx}(A^{\dagger}B)=\mu_{\vx}(B^{\dagger}A)=0,~~\forall
\vx \in \mathbb{Z}_4^T.
\eeea
Since $\ket{\psi}$ is arbitrary, we get
\beba
\tilde{A}^{\dagger}\tilde{A} \le I.
\label{eq:tildeAI}
\eeea
By Eqs.(\ref{eq:tildeAnorm}), (\ref{eq:tildeAI}) and lemma \ref{lem:svd}, we get there exists some $\tilde{V} \in \mathbb{U}_{2^{|T|}}$ such that 
\beba
D(\tilde{A},\tilde{V}) \le \sqrt{\delta/2},
\eeea 
and hence 
\beba
D(A, \tilde{V}_T \otimes I_{T^c}) \le \sqrt{\delta/2}.
\eeea
Meanwhile, we have
\beba
D(U,A) &\le \frac{1}{\sqrt{2N}} \|U-A\|\\
&=\frac{1}{\sqrt{2N}} \|B\|\\
&\le \sqrt{\delta/2}.
\eeea
As a result, 
\beba
D(U,\tilde{V}_T \otimes I_{T^c})
&\le D(U,A)+D(A, \tilde{V}_T \otimes I_{T^c})\\
&\le \sqrt{\delta/2}+\sqrt{\delta/2}\\
&=2\sqrt{\delta}.
\eeea
\end{proof}

Note that algorithm \ref{alg:junta} can also be efficiently implemented.

\section{Testing Clifford Group}

%introduce pauli group and clifford group
The Pauli group on $n$ qubits is defined as
the subgroup of $\UN$ generated by
$X_1 = X \otimes I \otimes \dots \otimes I,$
$Z_1 = Z \otimes I \otimes \dots \otimes I,$
$X_2 = I \otimes X \otimes \dots \otimes I,$
$Z_2 = I \otimes Z \otimes \dots \otimes I,$
$\dots,$
$X_n = I \otimes I \otimes \dots \otimes X,$
$Z_n = I \otimes I \otimes \dots \otimes Z.$
Equivalently, 
\beba
\Pn=\{i^k \sigma_{\vx}: k\in \mathbb{Z}_4, \vx \in \mathbb{Z}_4^n\}.
\eeea
The Clifford group on $n$ qubits is defined as
the normalizer of $\Pn$, i.e.
\beba
\Cn =\{U \in \UN: UhU^{\dagger} \in \Pn,~\forall h \in \Pn\}.
\eeea
Both Pauli group and Clifford group play important roles in quantum error correction \cite{Sho95,BDSW96,Ste96} and fault-tolerant quantum computation \cite{Sho96,Got98}.
 
Before stating our result about testing Clifford group, it is necessary to first present the following result about testing Pauli group.

%test pauli group
\begin{lemma}[implicit in Ref.\cite{MO09}]
The Pauli group $\Pn$ can be $\epsilon$-tested with query complexity $O({1}/{\epsilon^2})$.
\end{lemma}
\begin{proof}
Consider the following testing algorithm:

\begin{algorithm}
\caption{~~~Testing $\Pn$}
\begin{tabular}{p{0.15 \columnwidth}p{0.75 \columnwidth}}
\textbf{Input:} & $U \in \UN$ is given as a blackbox.
$\epsilon>0$ is a proximity parameter. \\
\textbf{Steps:} & \\
& Prepare $O({1}/{\epsilon^2})$ copies of $\ket{v(U)}$. Measure each copy in the basis $\{\ket{v(\sigma_{\vx})}\}_{\vx \in \mathbb{Z}_4^n}$. If all measurements get the same outcome, then accept. Otherwise, reject.
\end{tabular}
\label{alg:pn}
\end{algorithm}

Let $m=c/\epsilon^2$ be the number of copies of $\ket{v(U)}$ used in the algorithm. Since
\beba
\ket{v(U)} = \sum\limits_{\vx \in \mathbb{Z}_4^n}{\mu_{\vx}(U)\ket{v(x)}},
\eeea
the algorithm accepts $U$ with probability 
$\sum\limits_{\vx \in \mathbb{Z}_4^n}|\mu_{\vx}(U)|^{2m}$.
If $U \sim \sigma_{\vx_0}$ for some $\vx_0 \in \mathbb{Z}_4^n$, then $|\mu_{\vx_0}(U)|=1$ and $|\mu_{\vx}(U)|=0$ for any $\vx \neq \vx_0$, which implies that every measurement gets the outcome
$\vx_0$ and the algorithm accepts $U$ with certainty. On the other hand, if $D(U,\Pn)\ge\epsilon$, then by Eqs.(\ref{eq:dip}) and (\ref{eq:pdcof}) we get $|\mu_{\vx}(U)| \le 1-\epsilon^2$ for any $\vx
\in \mathbb{Z}_4^n$. Consequently,
\ben
\item if $0<\epsilon \le 1/\sqrt{2}$, then the probability of the algorithm accepting $U$ is
\beba
\sum\limits_{\vx \in \mathbb{Z}_4^n}|\mu_{\vx}(U)|^{2m}
&\le (1-\epsilon^2)^{2m}+ (\epsilon^2)^{2m} \\
&\le 2(1-\epsilon^2)^{2m} \\
&\le 2 e^{-2\epsilon^2 m} \\
&=2 e^{-2c}.
\eeea
By choosing a sufficiently large $c$, we can make this probability smaller than $1/3$. 
\item if $
1/\sqrt{2}<\epsilon<1$, then the probability of the algorithm accepting $U$ is
\beba
\sum\limits_{\vx \in \mathbb{Z}_4^n}|\mu_{\vx}(U)|^{2m}
&\le (1-\epsilon^2)^{2m}(\frac{1}{(1-\epsilon^2)^2}+1) \\
&\le 2(1-\epsilon^2)^{2m-2} \\
&\le (\frac{1}{2})^{2m-1}, 
\eeea
which is smaller than $1/3$ as long as $m \ge 2$. 
\een
Overall, this algorithm $\epsilon$-tests $\Pn$ with query complexity $O(1/\epsilon^2)$.
\end{proof}

%test clifford group
Now we turn to the testing of Clifford group. Note that if $U \in \Cn$, then $D(U\sigma_{\vx}U^{\dagger}, \Pn)=0$ for any $\vx \in \mathbb{Z}_4^n$. The following lemma shows that the converse is also true in the approximate sense.
\begin{lemma}
If $D(U\sigma_{\vx}U^{\dagger}, \Pn) \le \delta$ for any $\vx \in
\mathbb{Z}_4^n$, then $D(U,\Cn) \le 4\delta$.
\label{lem:clif}
\end{lemma}
\begin{proof}
See Appendix A.
\end{proof}
Therefore, in order to test whether $U$ is in $[\Cn]$ or far from it, it is sufficient to test whether $U \sigma_{\vx} U^{\dagger}$ is in $[\Pn]$ or far from it
for every $\vx \in \mathbb{Z}_4^n$. Note that  $\sigma_{\vx}$ can be generated by $X_1,Z_1,\dots,X_n,Z_n$. So $U \sigma_{\vx} U^{\dagger}$ can be generated by $UX_1U^{\dagger}, UZ_1U^{\dagger},\dots,UX_nU^{\dagger}, UZ_nU^{\dagger}$. For example, if $g=X_1Z_2X_nZ_n \in \Pn$, then we have $UgU^{\dagger}=(UX_1U^{\dagger})(UZ_2U^{\dagger})(UX_nU^{\dagger})(UZ_nU^{\dagger})$. This means that $U\sigma_{\vx}U^{\dagger}$ is close to $\Pn$ for any $\vx$ if and only if $UX_jU^{\dagger}$ and $UZ_jU^{\dagger}$ are close to $\Pn$ for every $j$. This suggests that we can apply the algorithm \ref{alg:pn} to each $UX_jU^{\dagger}$ and $UZ_jU^{\dagger}$, and accepts 
$U$ if and only if every subtest accepts with an appropriate proximity parameter. An algorithm based on a similar idea was given in Ref.\cite{Low09}.
However, this approach has the drawback that it needs to execute $2n$ subtests 
on $UX_jU^{\dagger}$'s and $UZ_jU^{\dagger}$'s, and consequently its query complexity depends on $n$. Here we present a better algorithm whose query complexity only depends on the proximity parameter $\epsilon$.

One basic idea is that we still test the distance between $U\sigma_{\vx}U^{\dagger}$ and $\Pn$, but we only do this test for a few random $\vx$'s. The key observation is that $U \Pn U^{\dagger}$ is a group. This group structure ensures that, $U\sigma_{\vx}U^{\dagger}$ is close to $\Pn$ for any $\vx$ if and only if $U\sigma_{\vx}U^{\dagger}$ is close to $\Pn$ for a sufficiently large fraction of $\vx$. Specifically,
\begin{lemma}
If at least $2/3$ fraction of $\vx \in \mathbb{Z}_4^n$ satisfies
$D(U\sigma_{\vx}U^{\dagger}, \Pn) \le \delta$, then for any $\vy \in \mathbb{Z}_4^n$, we have $D(U\sigma_{\vy}U^{\dagger}, \Pn) \le 2\delta$. 
\label{lem:good}
\end{lemma}
\begin{proof}
We prove this result by using the pigeonhole principle. A $\vx \in \mathbb{Z}_4^n$ is said to be \emph{$\delta$-good} if $D(U\sigma_{\vx}U^{\dagger}, \Pn) \le \delta$. Fix a $\vy \in \mathbb{Z}_4^n$. Divide $\mathbb{Z}_4^n$ into $4^n/2$ pairs: each pair consists of $\vx$ and $\vx \oplus \vy$ for some $\vx \in \mathbb{Z}_4^n$. Since at least $2/3$ fraction of $\mathbb{Z}_4^n$ is $\delta$-good, at least one pair contains two $\delta$-good vectors. Let $\vz$ and $\vz \oplus \vy$ be any such a pair. Then 
\beba
D(U\sigma_{\vy}U^{\dagger},\Pn)
&=&D(U\sigma_{\vz}U^{\dagger}\cdot
U(\sigma_{\vz}\sigma_{\vy})U^{\dagger}, \Pn)\\
&=&D(U\sigma_{\vz}U^{\dagger}\cdot U\sigma_{\vz \oplus \vy}U^{\dagger}, \Pn)\\
& \le & D(U\sigma_{\vz}U^{\dagger}, \Pn)+
D(U\sigma_{\vz \oplus \vy}U^{\dagger},\Pn)\\
& \le & 2\delta.
\eeea
\end{proof}
Therefore, we can run the algorithm \ref{alg:pn} on only a few random $\vx \in \mathbb{Z}_4^n$, and then estimate the fraction of good $\vx$'s with sufficiently good precision and hence estimate the distance between $U$ and $\Cn$. Now we give more details.

\begin{theorem}
If both $U$ and $U^{\dagger}$ can be accessed, then the Clifford group $\Cn$ can be $\epsilon$-tested with query complexity $O({1}/{\epsilon^2})$ 
\end{theorem}
\begin{proof}
Consider the following testing algorithm:

\begin{algorithm}
\caption{~~~Testing $\Cn$}
\begin{tabular}{p{0.15 \columnwidth}p{0.75 \columnwidth}}
\textbf{Input:} & $U,U^{\dagger} \in \UN$ are given as blackboxes.
$\epsilon>0$ is a proximity parameter. \\
\textbf{Steps:} & \\
& 1. Repeat the following procedure $O(1)$ times: \\
& \bit
\item Pick a $\vx \in \mathbb{Z}_4^n$ uniformly at random. 
\item Run the algorithm \ref{alg:pn} on $U\sigma_{\vx}U^{\dagger}$
with the proximity parameter $\epsilon/8$, where the oracle $U\sigma_{\vx}U^{\dagger}$ is simulated by concatenating $U^{\dagger}$, $\sigma_{\vx}$ and $U$. 
\eit\\
& 2. If all iterations accept, then accept. Otherwise, reject.
\end{tabular}
\label{alg:cn}
\end{algorithm}

%analysis
If $U \in [\Cn]$, then $U\sigma_{\vx}U^{\dagger} \in [\Pn]$ for any $\vx$ and hence the algorithm always accepts $U$. On the other hand, we will show that if the algorithm accepts $U$ with probability greater than $1/3$, then at least $2/3$ fraction of $\vx \in \mathbb{Z}_4^n$ is $\epsilon/8$-good, and then by lemma \ref{lem:clif} and lemma \ref{lem:good}, we get $D(U,\Cn) \le \epsilon$. Suppose the fraction of $\epsilon/8$-good $\vx$ is $p$. Then if we pick such a good $\vx$, then algorithm 3 accepts $U\sigma_{\vx}U^{\dagger}$ with probability at most $1$; otherwise, the algorithm 3 accepts $U\sigma_{\vx}U^{\dagger}$ with probability at most $1/3$.  Hence, the probability that algorithm \ref{alg:cn} accepts $U$ is at most $(p+(1-p)/3)^C$, where $C=O(1)$ is the number of iterations. If $p \le 2/3$, then this probability is at most $(7/9)^C$ which is smaller than $1/3$ as long as $C \ge 5$.
\end{proof}

Note that algorithm \ref{alg:pn} is efficiently implementable, and hence so is algorithm
\ref{alg:cn}.

\section{Testing Any Finite Subset}
So far we have studied the testing of several special subsets of $\UN$. In this section, we will present an algorithm that tests any finite subset of $\UN$ and thus also give an upper bound on the query complexity.

\begin{theorem}
Suppose $\S=\{W_1,W_2,\dots,W_M\} \subset \UN$.
If $\max\limits_{1\le i<j \le M} |\ip{W_i}{W_j}|=N(1-\delta)$,
then $\S$ can be $\epsilon$-tested with query complexity
$O(\dfrac{\log M}{\min\{\epsilon^2,\delta\}})$.
\label{thm:general}
\end{theorem}
\begin{proof}
Our basic idea is as follows. Suppose $U$ is either in $[\S]$ or far from $[\S]$. Then $\ket{v(U)}$ is either in the subspace spanned by $\ket{v(W_1)},\ket{v(W_2)}, \dots, \ket{v(W_M)}$, or has bounded projection onto this subspace. Since $|\bracket{\varphi^{\otimes k}}{\psi^{\otimes k}}|
=|\bracket{\varphi}{\psi}|^k$, we would expect that $\ket{v(U)}^{\otimes k}$ 
is either in the subspace spanned by $\ket{v(W_1)}^{\otimes k},
\ket{v(W_2)}^{\otimes k}, \dots, \ket{v(W_M)}^{\otimes k}$, or 
has exponentially small projection onto this subspace.
So by making a projective measurement on $\ket{v(U)}^{\otimes k}$
we should be able to distinguish the two cases. In the next, we are going to make this argument rigorous.

Consider the following testing algorithm:\\

\begin{algorithm}
\caption{~~~Testing $\S=\{W_1,W_2,\dots,W_M\} \subset \UN$}
\begin{tabular}{p{0.15 \columnwidth}p{0.75 \columnwidth}}
\textbf{Input:} & $U \in \UN$ is given as a blackbox.
$\epsilon>0$ is a proximity parameter. And suppose
$\delta=1-\frac{1}{N}\max\limits_{1 \le i<j \le M} |\ip{W_i}{W_j}|$.\\
\textbf{Steps:} & \\
& 1. Prepare $K=O(\frac{\log M}{\min\{\epsilon^2,\delta\}})$ copies of
the state $\ket{v(U)}$.\\
& 2. Perform the measurement $\{\Pi_K,I-\Pi_K\}$ on 
$\ket{v(U)}^{\otimes K}$, where $\Pi_K$ is the projection operator
onto the subspace $\WK=\span\{\ket{v(W_1)}^{\otimes K}, \ket{v(W_2)}^{\otimes K}, \dots, \ket{v(W_m)}^{\otimes K}\}$. 
If the outcome corresponds to $\Pi_K$, then accept; otherwise, reject.\\
\end{tabular}
\label{alg:general}
\end{algorithm}

Obviously, if $U \in [\S]$, then the algorithm always accepts. So it remains to show if $D(U,\S) \ge \epsilon$, then the algorithm accepts with probability at most $1/3$.

Choose $K=O(\frac{\log M}{\min\{\epsilon^2,\delta\}})$
such that 
\bea
(1-\delta)^K &\le e^{-\delta K} & \le \dfrac{1}{5M},
\label{eq:generalKd}\\
(1-\epsilon^2)^K & \le e^{-\epsilon^2 K} & \le \dfrac{1}{5M}.
\label{eq:generalKe}
\eea

Without loss of generality, we can assume that $\ip{U}{W_i}$ is real for all $i$, because, if otherwise, we can replace $W_i$ by some $W_i' \in [W_i]$ so that this condition is fulfilled. Then we have 
\beba
\bracket{v(U)}{v(W_i)}&=&\frac{1}{N}\ip{U}{W_i}\\
&=&1-D^2(U,W_i)\\
& \le & 1-\epsilon^2.
\label{eq:vuvwi}
\eeea
Moreover, without loss of generality, we can also assume that $[W_1],[W_2],\dots,[W_M]$ are disjoint. Then $\ket{v(W_1)},\ket{v(W_2)},\dots,\ket{v(W_M)}$
are linearly independent and $\WK$ is $M$-dimensional.
Let $\{\ket{\psi_i}= \ssll{j=1}{M} \lambda_{i,j} \ket{v(W_j)}^{\otimes K}\}_{i=1}^M$ be an
arbitrary orthonormal basis for $\WK$. Then we have
\beba
M &=& \ssll{i=1}{M}\bracket{\psi_i}{\psi_i}\\
&=& \ssll{i,j,j'=1}{M} \lambda_{i,j}^*\lambda_{i,j'} (\bracket{v(W_j)}{v(W_{j'})})^K\\
&=& \ssll{i=1}{M}\ssll{j=1}{M} |\lambda_{i,j}|^2
+\ssll{i=1}{M} \sum\limits_{j \neq j'} \lambda_{i,j}^*\lambda_{i,j'}(\bracket{v(W_j)}{v(W_{j'})})^K\\
& \ge & \ssll{i=1}{M}\ssll{j=1}{M} |\lambda_{i,j}|^2 - \ssll{i=1}{M}\sum\limits_{j \neq j'} \frac{|\lambda_{i,j}|^2+|\lambda_{i,j'}|^2}{2}(1-\delta)^K \\
& \ge & \frac{1}{2}\ssll{i=1}{M}\ssll{j=1}{M} |\lambda_{i,j}|^2,
\eeea
i.e.
\beba
\ssll{i=1}{M}\ssll{j=1}{M} |\lambda_{i,j}|^2 \le 2M.
\label{eq:lambdaij}
\eeea
Now if we perform the measurement $\{\Pi_K,I-\Pi_K\}$ on 
$\ket{v(U)}^{\otimes K}$, then the probability of
obtaining the outcome corresponding to $\Pi_K$ is
\beba
&&\tr(\Pi_K (|v(U) \rangle\langle v(U)|)^{\otimes K})\\
&=& \ssll{i=1}{M} |\langle \psi_i|v(U) \rangle^{\otimes K}|^2\\
&=& \ssll{i=1}{M} | \ssll{j=1}{M} \lambda_{i,j}^* (\langle v(W_j) |v(U) \rangle)^K|^2\\
&=& \ssll{i=1}{M} \ssll{j,j'=1}{M} \lambda_{i,j}^*\lambda_{i,j'} (\langle v(W_j) |v(U) \rangle)^K
(\langle v(U) |v(W_{j'}) \rangle)^K \\
&\le & (1-\epsilon^2)^{2K}  \ssll{i=1}{M} \ssll{j,j'=1}{M} |\lambda_{i,j}||\lambda_{i,j'}| \\
&\le & \frac{1}{25M^2}  \ssll{i=1}{M} \ssll{j,j'=1}{M} \frac{|\lambda_{i,j}|^2+|\lambda_{i,j'}|^2}{2} \\
&\le & \frac{1}{25M} \ssll{i=1}{M}\ssll{j=1}{M} |\lambda_{i,j}|^2 \\
&\le & \frac{1}{10},
\eeea
where in the fourth step we use Eq.(\ref{eq:vuvwi}),
and in the last step we use Eq.(\ref{eq:lambdaij}).
\end{proof}

\subsection{Example: Testing Permutations}
Let us demonstrate an application of theorem \ref{thm:general} to the permutation group. These operators just relabel the $n$ qubits of the system but do nothing else. Formally, let $S_n$ denote the group of permutations over $\{1,2,\dots,n\}$. Then any $\tau \in S_n$ is viewed as a unitary operator as follows
\beba
\tau \ket{i_1,i_2,\dots,i_n}=\ket{i_{\tau(1)}, i_{\tau(2)},\dots, i_{\tau(n)}}, 
\eeea 
for any $i_1,i_2, \dots, i_n=0,1.$

For any $\tau_1 \neq \tau_2$, let $\gamma = \tau_1^{-1}\tau_2 \neq I$. Note that $\gamma$ can be decomposed into several disjoint cycles
\beba
\gamma = (a_{1,1}, \dots, a_{1,k_1})
(a_{2,1}, \dots, a_{2,k_2})
\dots(a_{l,1}, \dots, a_{l,k_l}),
\label{eq:perm}
\eeea
where $l \le n-1$. So
\beba
\ip{\tau_1}{\tau_2}&=&\tr(\tau_1^{\dagger}\tau_2)\\
&=&\tr(\gamma)\\
&=&\tr(\sum\limits_{i_1,\dots,i_n=0,1} \ketbra{i_{\gamma(1)},\dots,i_{\gamma(n)}}{i_1,\dots,i_n})\\
&=&\sum\limits_{i_1,\dots,i_n=0,1} \bracket{i_1,\dots,i_n}{i_{\gamma(1)},\dots,i_{\gamma(n)}}.\\
\eeea
By Eq.(\ref{eq:perm}), the only nonzero terms on the right-hand side are those satisfying 
\beba
i_{a_{j,1}}=i_{a_{j,2}}=\dots=i_{a_{j,k_j}},~
\forall j=1,2,\dots,l.
\eeea
Hence 
\beba
\ip{\tau_1}{\tau_2}=2^l \le \dfrac{N}{2}.
\eeea
This holds for any $\tau_1 \neq \tau_2$. Therefore,
\beba
\delta=1-\dfrac{1}{N}\max\limits_{\tau_1 \neq \tau_2}|\ip{\tau_1}{\tau_2}| \ge \dfrac{1}{2}.
\eeea
Besides, note that $|S_n|=n!=O(e^{n \log n})$. So by theorem \ref{thm:general}, we get
\begin{theorem}
The permutation group $S_n$ can be $\epsilon$-tested
with query complexity $O({n}\log n/{\epsilon^2})$.
\end{theorem}

Note that it is unknown whether algorithm \ref{alg:general} can be implemented efficiently in general. So it still remains open to find a testing algorithm for $S_n$ that has both polynomial query complexity and polynomial computational complexity.

\section{Conclusion}
%summary
In summary, we have systematically studied property testing of unitary operators with respect to the distance measure $D$ that reflects the average difference between unitary operators. We present efficient algorithms for testing the orthogonal group, quantum juntas and Clifford group. All these algorithms have only one-sided error and their query complexities are independent of the system's size. We also give an algorithm that tests any finite subset of the unitary group, and show an application of this algorithm to the permutation group. This algorithm also has one-sided error and polynomial query complexity, although it is unknown whether it can be efficiently implemented in general. 

%open questions
Despite the progress made in recent papers and ours, the testing of quantum objects (states or operations) still remains widely open. We hope that our work can shed some light on this topic and stimulate further research. Here are several directions that seems particularly interesting to us:

First, in this paper we focus on giving \textit{upper bounds} on the query complexity of testing unitary operators. It is worth developing powerful techniques that can derive \textit{lower bounds} on the query complexity of the same task.
In particular, it is interesting to prove
that our testing algorithms are optimal, or give
better algorithms.

Second, as mentioned above, we do not know whether algorithm \ref{alg:general} can be efficiently implemented, or whether a better algorithm can be given for testing general discrete subsets. And it would be interesting to give an efficient algorithm for testing the permutation group. Furthermore, can we give some general results on testing continuous subgroups?

Third, here we only considered testing unitary operators. It is worth exploring property testing of general quantum states and quantum operations. For example, is it possible to test (or estimate) the entangling power of multipartite quantum operations by using only a few queries? Is it possible to test whether a given multipartite state belongs to an interesting class, such as symmetric or antisymmetric states, by consuming only a few copies?

At last, the query complexity of property testing crucially depends on the distance measure used. Here we considered a particular distance measure $D$ that seems quite suitable for comparing unitary operators. It would be interesting to study the testing of quantum states or channels with respect to other natural distance measures.

\section*{Acknowledgments}
This research was supported by NSF Grant CCR-0905626
and ARO Grant W911NF-09-1-0440. 

\begin{appendix}
\section{Proof of lemma \ref{lem:clif}}
If $\delta \ge 1/4$, then
obviously $D(U,\Cn)\le 1 \le 4\delta$.
So from now on we assume $\delta<1/4.$
Define $F: \mathbb{Z}_4^n \to \mathbb{Z}_4^n$
and $\Theta: \mathbb{Z}_4^n \to [0,2\pi)$
such that for any $\vx \in \mathbb{Z}_4^n$,
\beba
D(U\sigma_{\vx}U^{\dagger},
\Pn)=\frac{1}{\sqrt{2N}}\| U\sigma_{\vx}U^{\dagger}-e^{i\Theta(\vx)} \sigma_{F(\vx)} \| \le \delta,
\label{eq:uxu}
\eeea 
i.e.  $e^{i\Theta(\vx)}\sigma_{F(\vx)}$ is
the closest element to $U\sigma_{\vx}U^{\dagger}$
in $[\Pn]$. We will prove that
\ben
\item $\Theta(\vx)=0$ or $\pi$, $\forall \vx \in \mathbb{Z}_4^n$;
\item $F(\vx)\oplus F(\vy)=F(\vx \oplus \vy)$, $\forall \vx,\vy \in \mathbb{Z}_4^n$;
\item $\sigma_{\vx}$ and $\sigma_{\vy}$
commute (or anticommute) if and only
if $\sigma_{F(\vx)}$ and $\sigma_{F(\vy)}$
commute (or anticommute),
$\forall \vx,\vy \in \mathbb{Z}_4^n$.
\een
Let us first assume these statements are true. Then the mapping $\sigma_{\vx} \to \gamma_{\vx}\sigma_{F(\vx)}$ is an isomorphism, where $\gamma_{\vx}=e^{i\Theta(\vx)}= \pm 1$.
Hence there exists a Clifford operator $C \in \Cn$ such that 
\beba
C \sigma_{\vx} C^{\dagger}= \gamma_{\vx}\sigma_{F(\vx)},~\forall \vx\in \mathbb{Z}_4^n.
\eeea
Then by Eq.(\ref{eq:uxu}), we have
\beba
\langle U\sigma_{\vx}U^{\dagger},
C \sigma_{\vx} C^{\dagger}
\rangle \ge N(1-\delta^2),
~\forall \vx\in \mathbb{Z}_4^n.
\eeea
Consequently,
\beba
1-\delta^2
&\le & \frac{1}{N^3} \sum\limits_{\vx \in \mathbb{Z}_4^n} \langle U\sigma_{\vx}U^{\dagger},
C \sigma_{\vx} C^{\dagger} \rangle \\
&\le & \frac{1}{N^3} \sum\limits_{\vx \in \mathbb{Z}_4^n}  
\tr(U\sigma_{\vx}U^{\dagger}C \sigma_{\vx} C^{\dagger}) \\
&\le & \frac{1}{N^3} 
\tr(U
(\sum\limits_{\vx \in \mathbb{Z}_4^n} 
\sigma_{\vx} U^{\dagger}C \sigma_{\vx}) C^{\dagger}) \\
&=& \frac{1}{N^3} 
\tr(U
(N\tr(U^{\dagger}C)I) C^{\dagger}) \\
&=& \frac{1}{N^2} 
|\tr(U^{\dagger}C)|^2, 
\eeea
where in the fourth step we use the fact 
\beba
\frac{1}{N}\sum\limits_{\vx \in \mathbb{Z}_4^n} 
\sigma_{\vx} A \sigma_{\vx}
= \tr(A) I, ~\forall A \in \MNN.
\eeea
Hence, we have 
\beba
|\langle U,C\rangle| \ge N \sqrt{1-\delta^2}
\ge N (1-\delta^2).
\eeea
Then by Eq.(\ref{eq:dip}) we obtain
\beba
D(U,C) \le \delta.
\eeea

Now we prove statements 1-3. The first one is easy. 
Since $\sigma_{\vx}$ is Hermitian,
$U\sigma_{\vx}U^{\dagger}$ is also Hermitian.
As a result, $\tr(U\sigma_{\vx}U^{\dagger}\sigma_{F(\vx)})$ is real. $\Theta(\vx)$ is chosen
from $[0,2\pi)$ such that 
\beba
e^{i\Theta(\vx)}\tr(U\sigma_{\vx}U^{\dagger}\sigma_{F(\vx)})
=|\tr(U\sigma_{\vx}U^{\dagger}\sigma_{F(\vx)})|.
\eeea 
So $e^{i\Theta(\vx)}$ can only be $\pm 1$ and
$\Theta(\vx)$ can only be $0$ or $\pi$.

To prove the second statement, let us suppose
\beba
U \sigma_{\vx} U^{\dagger} &=& \lambda_{F(\vx)} \sigma_{F(\vx)}
+ \sum\limits_{\vz \neq F(\vx)} \lambda_{\vz} \sigma_{\vz}, \\
U \sigma_{\vy} U^{\dagger} &=& \xi_{F(\vy)} \sigma_{F(\vy)}
+ \sum\limits_{\vz \neq F(\vy)} \xi_{\vz} \sigma_{\vz}.\\
\label{eq:usxuusyu}
\eeea
By $D(U \sigma_{\vx} U^{\dagger},\sigma_{F(\vx)}) \le \delta$
and $D(U \sigma_{\vy} U^{\dagger},\sigma_{F(\vy)}) \le \delta$,
we get
\beba
|\lambda_{F(\vx)}| \ge 1-\delta^2,\\
|\xi_{F(\vy)}| \ge 1-\delta^2.\\
\eeea
Now 
\beba
U \sigma_{\vx \oplus \vy}U^{\dagger}
\sim
U \sigma_{\vx} U^{\dagger} \cdot
U \sigma_{\vy} U^{\dagger}
= \sum\limits_{\vz \in \mathbb{Z}_4^n} \eta_{\vz} \sigma_{\vz}.
\eeea
for some $\eta_{\vz}$'s. Since $D(U \sigma_{\vx \oplus \vy}U^{\dagger},
\sigma_{F(\vx \oplus \vy)}) \le \delta$,
we have 
\beba
|\eta_{F(\vx \oplus \vy)}| \ge 1-\delta^2.
\eeea
Meanwhile, the coefficient $\eta_{F(\vx)\oplus F(\vy)}$ satisfies
\beba
|\eta_{F(\vx)\oplus F(\vy)}|
&=&|\sum\limits_{\vz_1\oplus\vz_2=F(\vx)\oplus F(\vy)}
i^{\vz_1 \odot \vz_2} 
\lambda_{F(\vz_1)}\xi_{F(\vz_2)}|\\
&=&|i^{F(\vx) \odot F(\vy)} 
\lambda_{F(\vx)}\xi_{F(\vy)}\\
&&+\sum\limits_{\vz\neq F(\vx)}
i^{\vz \odot (\vz \oplus F(\vx) \oplus F(\vy))} 
\lambda_{\vz}\xi_{\vz \oplus F(\vx) \oplus F(\vy)}|\\
&\ge &
|\lambda_{F(\vx)}\xi_{F(\vy)}|
-\sum\limits_{\vz\neq F(\vx)}
|\lambda_{\vz}\xi_{\vz \oplus F(\vx) \oplus F(\vy)}|\\
&\ge & (1-\delta^2)^2
-\sum\limits_{\vz\neq F(\vx)}
\frac{|\lambda_{\vz}|^2+|\xi_{\vz \oplus F(\vx) \oplus F(\vy)}|^2}{2} \\
&=& (1-\delta^2)^2
-\frac{1-|\lambda_{F(\vx)}|^2}{2}
-\frac{1-|\xi_{F(\vy)}|^2}{2} \\
&\ge & (1-\delta^2)^2-(1-(1-\delta^2)^2)\\
&\ge & 1-4 \delta^2.
\label{eq:etafxfy}
\eeea
If $F(\vx) \oplus F(\vy) \neq F(\vx \oplus \vy)$,
then we would have
\beba
1 &\ge& |\eta_{F(\vx \oplus \vy)}|^2
+ |\eta_{F(\vx)\oplus F(\vy)}|^2 \\
&\ge & (1-\delta^2)^2+(1-4\delta^2)^2 \\
& \ge & 2-10 \delta^2 \\
& > & 1,
\eeea
which is a contradiction.
Therefore, $F(\vx) \oplus F(\vy)=F(\vx \oplus \vy)$.

Finally, we prove statement 3. Let us first consider the case when $\sigma_{\vx}$ and $\sigma_{\vy}$ commute. Then 
$U\sigma_{\vx}U^{\dagger}$ and 
$U\sigma_{\vy}U^{\dagger}$ also commute.
We still assume Eq.(\ref{eq:usxuusyu}),
and suppose 
\beba
&U \sigma_{\vx} U^{\dagger} \cdot
U \sigma_{\vy} U^{\dagger}
&= \sum\limits_{\vz \in \mathbb{Z}_4^n} \eta_{\vz} \sigma_{\vz} \\
=& U \sigma_{\vy} U^{\dagger} \cdot
U \sigma_{\vx} U^{\dagger}
&= \sum\limits_{\vz \in \mathbb{Z}_4^n} \chi_{\vz} \sigma_{\vz}.
\eeea
If $\sigma_{F(\vx)}$ and
$\sigma_{F(\vy)}$ anticommute, then we have
$i^{F(\vx) \odot F(\vy)}=-i^{F(\vy) \odot F(\vx)}$, 
and furthermore, by Eq.(\ref{eq:etafxfy}), we get 
\beba
0&=&|\eta_{F(\vx)\oplus F(\vy)}
-\chi_{F(\vy)\oplus F(\vx)}|\\
&=&|(i^{F(\vx) \odot F(\vy)}- i^{F(\vy) \odot F(\vx)})
\lambda_{F(\vx)}\xi_{F(\vy)}\\
&&+\sum\limits_{\vz\neq F(\vx)}
(i^{\vz \odot (\vz \oplus F(\vx) \oplus F(\vy))} 
-i^{(\vz \oplus F(\vx) \oplus F(\vy)) \odot \vz} )\\
&&\lambda_{\vz}\xi_{\vz \oplus F(\vx) \oplus F(\vy)}| \\
& \ge & 2|\lambda_{F(\vx)}\xi_{F(\vy)}|
-2\sum\limits_{\vz\neq F(\vx)}
|\lambda_{\vz}\xi_{\vz \oplus F(\vx) \oplus F(\vy)}| \\
& \ge & 2(1-4\delta^2) \\
& > & 0,
\eeea
which is a contradiction. So $\sigma_{F(\vx)}$ and $\sigma_{F(\vy)}$ must commute. The case when $\sigma_{\vx}$ and $\sigma_{\vy}$ anticommute can be handled by a similar argument.
\end{appendix}

\end{document}